\documentclass{amsart}

\newtheorem{theorem}{Theorem}

\newtheorem{proposition}[theorem]{Proposition}
\theoremstyle{remark}

\begin{document}

\title[On the number of electrons that a nucleus can bind]{On the number of electrons\\ that a nucleus can bind}

\thanks{Contribution to the Proceedings of ICMP12, Aalborg, Denmark, August 6--11, 2012. \\ \copyright\, 2012 by
  the author. This paper may be reproduced, in its entirety, for
  non-commercial purposes.}
  
\author[P.~T. Nam]{Phan Th\`anh Nam}

\address{CNRS \& Universit\'e de Cergy-Pontoise, Department of Mathematics (UMR 8088), 2 avenue Adolphe Chauvin, 95302 Cergy-Pontoise, France. } 

\email{Phan-Thanh.Nam@u-cergy.fr }

\begin{abstract} We review some results on the ionization conjecture, which says that a neutral atom can bind at most one or two extra electrons. 
\end{abstract}
\date{{7 December, 2012}}
\maketitle

\section{Introduction}\label{aba:sec1}

It is an experimental fact that a neutral atom (in the vacuum) can bind at most one or two extra electrons. Although this fact is well-known by most physicists and chemists, providing a rigorous explanation for it using Schr\"odinger's model in quantum mechanics is a very challenging problem, often referred to as the {\it ionization conjecture} (see, e.g., \cite{Li84,So91,Si00,So03,LS09}). In this short paper we will review some results on this problem.

To be precise, an atom is a system with a classical nucleus, which is of charge $Z$ and fixed at the origin in $\mathbb{R}^3$, and $N$ non-relativistic quantum electrons. In atomic units, the $N$-electron system is described by the Hamiltonian
\begin{equation} \label{eq:non-rel-Hamiltonian}
H_{N,Z} = \sum\limits_{i = 1}^N {\left( { - \frac{1}{2}\Delta _i  - \frac{Z}
{{|x_i |}}} \right)}  + \sum\limits_{1 \le i < j \le N} {\frac{1}
{{|x_i  - x_j |}}} 
\end{equation}
acting on the space $\mathop  \bigwedge \limits_{i = 1}^N (L^2 (\mathbb{R}^3 ) \otimes \mathbb{C}^2)$ of square-integrable functions $\Psi\in L^2((\mathbb{R}^3\times \mathbb{C}^2) ^N)$ which are anti-symmetric with respect to exchanges of the variables, namely
$$\Psi(\ldots, y_i,\ldots,y_j,\ldots)= -\Psi(\ldots,y_j,\ldots,y_i,\ldots)$$
where $y_i$ and $y_j\in \mathbb{R}^3\times \mathbb{C}^2$ are the coordinates of corresponding electrons. The requirement of the anti-symmetry is to take Pauli's exclusion principle into account. Here we are taking the physical spin number $q=2$, but any other fixed number does not change our mathematical arguments presented below. The nuclear charge $Z$ is an integer in the physical case, but in our discussion it can be any positive number as well.  

The ground state energy of the system is  
\[
E(N,Z) = \inf {\text{spec}}~ H_{N,Z}=\mathop {\inf }\limits_{||\Psi ||_{L^2}  = 1} \langle \Psi ,H_{N,Z} \Psi \rangle.
\]
If the system has a ground state $\Psi$, then it must satisfy Schr\"odinger's equation
\begin{equation} \label{eq:Schrodinger-equation}
(H_{N,Z}-E(N,Z)) \Psi = 0.
\end{equation}
The $N$ electrons are said to be {\it bound} if $E(N,Z)<E(N-1,Z)$, namely one cannot remove any electron without paying some positive energy. Note that 
$$\text{ess spec}~H_{N,Z}=[E(N - 1),\infty)$$
due to the celebrated HVZ theorem \cite{Hu66,Wi64,Zh60}. Therefore, if the binding inequality $E(N,Z)<E(N-1,Z)$ holds true, then $E(N,Z)$ is an isolated eigenvalue of $H_{N,Z}$. We also remark that if $E(N,Z)=E(N-1,Z)$, then in principle $E(N,Z)$ may be also an eigenvalue of $H_{N,Z}$, but the corresponding eigenstates would be highly unstable since some electron can escape at infinity. 

Zhislin \cite{Zh60} showed that the binding always occurs when $N<Z+1$, namely the positive ions and neutral atoms do exist. In the following, let us denote by $N_c=N_c(Z)$ the largest number of electrons that can be bound by a nucleus of charge $Z$. It was first proved by Ruskai \cite{Ru82} and Sigal \cite{Si82} that $N_c$ is finite. 

The {\it ionization conjecture} can be now stated that $N_c\le Z+1$, or possibly $N_c\le Z+2$, for all $Z$. It is well-known that the hydrogen ion $H^{-}$ does exist, but nobody knows (both experimentally and theoretically) if some atom with $N_c=Z+2$ exists or not, see \cite[Sec. 12.3]{LS09} for further discussion.

We remark that the Pauli's exclusion principle plays an essential role in the ionization conjecture because for bosonic atoms, namely the atoms with ``bosonic" electron, it was shown that $\lim_{Z\to \infty}N_c/Z =t_c \approx 1.21$ by Benguria and Lieb \cite{BL83} (lower bound) and Solovej \cite{So90} (upper bound), where the numerical value $1.21$ is taken from \cite{Ba84}. 

The ionization problem was also considered in other models, such as the Thomas-Fermi and related theories \cite{Be79,Li81,BL85}, the Hartree-Fock theory \cite{So91,So03,DS10}, atoms in a magnetic field \cite{LSY94I,LSY94II,BR99,Se01,Nam-12,HS12}, atoms with pseudo-relativistic kinetic energy \cite{Li84,DS10,Nam-12}, and atoms in the polarized 
Dirac vacuum \cite{GLS11}. However, in this review we shall mainly restrict ourself to the non-relativistic case in (\ref{eq:non-rel-Hamiltonian}).

In the next sections we shall discuss two main approaches to the ionization problem. The first is the geometric method \cite{Si82,Ru82,Si84,LSST88}, which uses an appropriate localization to transfer the quantum problem to a problem of classical particles. The second is the PDE method \cite{Li84,FS90,SSS90,Nam-12}, which extracts information directly from Schr\"odinger's equation. 


\section{Geometric method}\label{sec:geometric-method}

If we describe the electrons as classical points in $\mathbb{R}^3$, then the binding can be understood as that the energy contributed by any electron is always negative. On the other hand, when the $N$-th electron runs to infinity (and the others stay bounded), its contribution to the total energy is
$$
- \frac{Z}{|x_N|} + \sum\limits_{i = 1}^{N-1} \frac{1}{|x_i - x_N|} \approx \frac{-Z+N-1}{|x_N|}.
$$
Thus we can deduce heuristically that $N\le Z+1$. 

More rigorously, for every configuration $\{x_i\}_{i=1}^N \subset \mathbb{R}^3$, if $N>2Z+1$, then the energy contributed by the farthest electron, $x_N$ says, is always positive because
\begin{equation} \label{eq:Sigal-classical-bound}
- \frac{Z}{|x_N|} + \sum\limits_{i = 1}^{N-1} \frac{1}{|x_i - x_N|} \ge  - \frac{Z}{|x_N|} + \sum\limits_{i = 1}^{N-1} \frac{1}{2|x_N|} = \frac{-2Z+N-1}{2|x_N|}>0.
\end{equation}
Using this observation and an appropriate localization, Sigal \cite{Si84} showed the upper bound $\limsup_{Z\to \infty}N_c/Z\le 2$. 

Later, Lieb, Sigal, Simon and Thirring \cite{LSST88} proved the {\it asymptotic neutrality} $\lim_{Z\to \infty} {N_c}/{Z}=1$ by using the following refined version of (\ref{eq:Sigal-classical-bound}).
\begin{proposition}\label{eq:LSST-classical-bound} For every $\varepsilon>0$, there exists $N_\varepsilon>0$ such that if $N\ge N_\varepsilon$, then for every subset $\{x_i\}_{i=1}^N \subset \mathbb{R}^3$, we have
$$
\mathop {\max }\limits_{1\le j\le N} \left\{ {\sum\limits_{1 \le i \le N,i \ne j} {\frac{1}
{{|x_i  - x_j |}}}  - \frac{{N(1 - \varepsilon )}}
{{|x_j |}}} \right\} \ge 0.
$$
\end{proposition}
Heuristically, this bound implies the no-binding when $(1-\varepsilon)N>Z$ and $N$ is large enough, because among arbitrary $N$ electrons, we can always remove {\it some} electron without increasing the total energy.

\section{PDE method}\label{sec:PDE-method}

By using Schr\"odinger equation (\ref{eq:Schrodinger-equation}), Lieb \cite{Li84} proved the uniform bound $N_c<2Z+1$ for all $Z>0$, which in particular settles the ionization conjecture for hydrogen. By combining Lieb's argument and the approximations to the Thomas-Fermi theory \cite{LS77}, Fefferman-Seco \cite{FS90} and Seco-Sigal-Solovej \cite{SSS90} showed that $N_c\le Z+CZ^{5/7}$, which is the best-known bound for large atoms. For small atoms, by modifying Lieb's argument, we obtained the following result \cite{Nam-12}.
\begin{theorem}\label{thm:new-bound} For every $Z > 0$, if the Schr\"odinger equation (\ref{eq:Schrodinger-equation}) has a solution, then either $N=1$ or $N<1.22\,Z+3Z^{1/3}.$
\end{theorem}

Our bound improves Lieb's bound when $Z\ge 6$. Before sketching the proof of Theorem \ref{thm:new-bound}, let us recall Lieb's proof in \cite{Li84} which contains three main ingredients. 

1. We multiply the Schr\"odinger equation (\ref{eq:Schrodinger-equation}) with $|x_N| \overline{\Psi}$ and take the integral. The idea goes back to Benguria's work on the Thomas-Fermi-von Weizs\"acker model \cite{Be79,Li81}. Then we use the inequality $H_{N-1,Z}\ge E(N,Z)$ on the subspace of the first $(N-1)$ electrons to remove the part irrelevant to the $N$-th variable.

2. The kinetic energy error is positive and can be ignored, since $(-\Delta)|x|+|x|(-\Delta)\ge 0$ in $L^2(\mathbb{R}^3)$, which is equivalent to Hardy's inequality $|x|^{-2}\le 4(-\Delta)$. 

3. Using the symmetry of $|\Psi|^2$ and the triangle inequality we have
\begin{equation} \label{eq:Lieb-triangle-inequality}
\sum_{i=1}^N \left\langle \Psi, \frac{|x_N|}{|x_i-x_N|} \Psi \right\rangle = \frac{1}{2}\sum_{i=1}^N \left\langle \Psi, \frac{|x_N|+|x_i|}{|x_i-x_N|} \Psi \right\rangle \ge \frac{N-1}{2}.
\end{equation} 

\begin{proof}[Proof of Theorem \ref{thm:new-bound}]
The proof is based on the idea that we can improve the factor $1/2$ in (\ref{eq:Lieb-triangle-inequality}) by multiplying Schr\"odinger's equation (\ref{eq:Schrodinger-equation}) with $|x_N|^2 \overline{\Psi}$, instead of $|x_N|\overline{\Psi}$. Rigorously, we have the following bound \cite[Prop. 1]{Nam-12} 
$$
\beta \ge \mathop {\inf }\limits_{\{x_i\}_{i=1}^N\subset \mathbb{R}^3 } \frac{{\sum\limits_{1 \le i < j \le N} {\frac{{|x_i |^2  + |x_j |^2 }}
{{|x_i  - x_j |}}} }}
{{(N-1)\sum\limits_{i = 1}^N {|x_i |} }} 
  \ge \beta -1.55\,N^{-2/3}
$$
where 
$$
\beta:= \inf_{\substack{\rm \rho ~probability\\ \rm ~measure~in \mathbb{R}^3}} \left\{ {\frac{{\iint\limits_{\mathbb{R}^3  \times \mathbb{R}^3 } {\frac{{x^2+y^2}}
{{2|x - y|}} {{\rm d}\rho} (x){{\rm d}\rho} (y)}}}
{{\int\limits_{\mathbb{R}^3 } {|x|  {{\rm d}\rho} (x)} } }} \right\} \ge 0.82.
$$
The lower bound $\beta \ge 0.82$ can be proved using two inequalities:
$$
\iint\limits_{\mathbb{R}^3  \times \mathbb{R}^3 } {\frac{{x^2  + y^2 }}
{{|x - y|}} {\rm d} \rho (x)\operatorname{d \rho} (y)} \ge \iint\limits_{\mathbb{R}^3  \times \mathbb{R}^3 } {\left( {|x - y| + \frac{2}
{3}\frac{{(\min \{ |x|,|y|\} )^2 }}
{{\max \{ |x|,|y|\} }}} \right)\operatorname{d \rho} (x){\rm d} \rho (y)},
$$
and
$$
\iint\limits_{\mathbb{R}^3  \times \mathbb{R}^3 }  {\frac{{x^2  + y^2 }}
{{|x - y|}}{\rm d}\rho (x)\operatorname{d \rho} (y)} \ge \iint \limits_{\mathbb{R}^3  \times \mathbb{R}^3 }{\left( \max\{|x|,|y|\}  + \frac{\min\{|x|,|y|\}^2}{|x-y| } \right)\operatorname{d \rho} (x){\rm d} \rho (y)}.
$$
The first inequality follows from the positivity of the Coulomb kernel $|x-y|^{-1}$ in $L^2(\mathbb{R}^3)$, and the second is a consequence of the classical bound in Proposition \ref{eq:LSST-classical-bound}.

In our proof of Theorem \ref{thm:new-bound}, we have to deal with the kinetic energy error, since we only have 
$$(-\Delta)|x|^2 + |x|^2(-\Delta) \ge -3/2 \:\:\:\:\:{\rm in}~L^2(\mathbb{R}^3)$$
and the constant $3/2$ is sharp. To control the error, we use the following bound (see \cite[Lemma 2]{Nam-12})
$$\left\langle {\Psi  ,|x_N |\Psi } \right\rangle \ge 0.553~Z^{-1}N^{2/3}.$$
which can be proved using the Lieb-Thirring inequality \cite{LT75} (with the improved constant in \cite{DLL08}). This is the place we need Pauli's exclusion principle and our bound should be read as $N_c\le 1.22\,Z+Cq^{2/3}Z^{1/3}$ where $q$ is the spin number. While Lieb's upper bound $N_c<2Z+1$ holds for both fermions and bosons, our bound eventually becomes $N_c\le 2.4\,Z$ for bosonic atoms.
\end{proof}

I believe that the above argument can be modified to give an alternative proof of the asymptotic neutrality. In the following, let me demonstrate the idea on the Thomas-Fermi theory, since it is still not clear for me how to do with Schr\"odinger's model. 

Let us consider the Thomas-Fermi equation \cite{LS77}
\begin{equation} \label{eq:TF-equation}
\gamma\rho(x)^{2/3} = \left[ Z|x|^{-1}-(\rho*|\cdot |^{-1}(x) -\mu \right]_+,
\end{equation}
where $\gamma>0$, $\mu\ge 0$, and $\rho\in L^1(\mathbb{R}^3)$ is a non-negative, radially symmetric function. 

To show that $\int \rho \le Z$, let us multiply the equation (\ref{eq:TF-equation}) with $|x|^k \rho(x)$ and integrate on $\{|x|\le R\}$. By Newton's Theorem, we have
\[
(\rho *|\cdot |^{ - 1})(x) = \int_{{\mathbb{R}^3}} {\frac{{\rho (y)dy}}{{|x - y|}}}  = \int_{{\mathbb{R}^3}} {\frac{{\rho (y) {\rm d}y}}{{\max \{ |x|,|y|\} }}} 
\]
Using the elementary inequality 
\[
\frac{{|x{|^k} + |y{|^k}}}{{\max \{ |x|,|y|\} }} \ge \left( {1 - \frac{1}{k}} \right)\left( {|x{|^{k - 1}} + |y{|^{k - 1}}} \right)
\]
we can conclude that 
$$\left(1-\frac{1}{k}\right){\int_{|y|\le R} \rho(y) {\rm d}y \le Z}.$$
By taking $k\to \infty$ and $R\to \infty$, we obtain $\int\rho \le Z$.

\section{Open problems}\label{sec:open-problems}

To conclude, let us mention some open problems related to the ionization conjecture. 

1. The uniform bound $N_c\le Z+C$, for some universal constant $C$, is still {unproven} \cite{LS09}. This bound was already proved in the Hartree-Fock theory \cite{So03}. In view of Proposition \ref{eq:LSST-classical-bound}, we can also consider the classical version of the ionization problem: ``Does there exist a constant $C$ such that for every subset $\{x_i\}_{i=1}^N \subset \mathbb{R}^3$,
\begin{equation} \label{eq:classicla-version}
\mathop {\max }\limits_{1\le j\le N} \left\{ \sum\limits_{1 \le i \le N,i \ne j} {\frac{1}
{{|x_i  - x_j |}}}  - \frac{N-C}{{|x_j |} } \right\} \ge 0\:\:? "
\end{equation}
In one dimension, we can show that (\ref{eq:classicla-version}) holds true with $C=1$. However, in higher dimensions, the problem is much more difficult.    
\smallskip

2. It was proved recently in \cite{LL12} that if the Hamiltonian $H_{N,Z}$ has {\it an}  eigenvalue, then $N<4Z+1$. This work raises the following question: ``Is it possible that $H_{N,Z}$ has no isolated eigenvalue, but has some embedded eigenvalue?" Moreover, the upper bound $N<4Z+1$ may be not optimal. Can we improve it?
\smallskip

3. The ionization conjecture is closely related to the question on the {\it stability of the radii} of atoms. To be precise, let $\Psi_{Z}$ be a ground state for $H_{Z,Z}$ and define the radius $R_Z$ by $\int_{|x|>R_Z} \rho_{\Psi_Z}(x) {\rm d}x=1/2$, where $\rho_{\Psi_Z}$ is the density,
\[
\rho _{\Psi _{\Psi_{Z}} } (x): = N\sum_{\sigma_1=1,2}\ldots\sum_{\sigma_N=1,2} \int\limits_{\mathbb{R}^{3(N - 1)} } {|\Psi_{Z} (x,\sigma_1;x_2,\sigma_2;\ldots;x_N,\sigma_N)|^2 {\rm d} x_2\ldots{\rm d} x_N } .
\]
It is {conjectured} that $C_0<R_Z<C_1$ for two universal constants $C_0$ and $C_1$ \cite{LS09}. Although this bound was already proved in the Hartree-Fock theory \cite{So03}, in 
Schr\"odinger's theory it is only known that $R_Z \ge CZ^{-5/21}$ \cite{SSS90}. On the other hand, for atoms restricted to two dimensions, we have $R_Z\to \infty$ as $Z\to \infty$ \cite{NPS-12}.
\smallskip

4. It is {conjectured} that the ionization energy $I(Z):=E(Z-1,Z)-E(Z,Z)$ is bounded independently of $Z$ \cite{SSS90,So03}. The best-known result in Schr\"odinger's theory is that $I(Z)\le CZ^{20/21}$ \cite{SSS90}, and the uniform bound $I_Z\le C$ was already proved in the Hartree-Fock theory \cite{So03}. Note also that $I(Z)$ if of order $Z^2$ for bosonic atoms \cite{Ba91}.
\smallskip

5. It is {conjectured} that the ground state energy $E(N,Z)$ is a convex function in $N$, for every $Z$ fixed (some people would also refer to this problem as the {\em ionization conjecture}). Although this convexity conjecture is rather strong, the following consequence seems to be very reasonable: ``If a nucleus of charge $Z$ can bind $N$ electrons, then it can also bind $N-1$ electrons" \cite{LS09}. But even that binding property is still an {open} problem.  

\section*{Acknowledgments}

I am grateful to Mathieu Lewin and Jan Philip Solovej for various helpful discussions, and to Julien Sabin and the anonymous referee for valuable comments. Financial support from the European Research Council under the European Community's Seventh Framework Programme (FP7/2007-2013 Grant Agreement MNIQS 258023) is gratefully acknowledged.

\bibliographystyle{ws-procs975x65}

\end{document}